\documentclass[runningheads]{llncs}
\usepackage{amssymb,amsfonts,amsmath}
\usepackage{booktabs,tabularx,cite}
\usepackage{url,xspace,soul,wrapfig}
\usepackage{graphicx,paralist,enumerate}
\usepackage[hidelinks]{hyperref}
\usepackage{subfigure}
\usepackage[textsize=footnotesize]{todonotes}
\usepackage[basic]{complexity}
\usepackage[capitalise]{cleveref}
\newcommand{\bigoh}{\mathcal{O}}

\newcommand{\clos}{\ensuremath{\text{clos}}}
\newcommand{\td}{td}

\newcommand{\qnp}{\textsc{Queue Number}\xspace}

\newtheorem{observation}{Observation}

\DeclareMathOperator{\qn}{qn}

\title{Parameterized Algorithms\\ for Queue Layouts\thanks{Research of FM partially supported by Dip. Ingegneria Univ. Perugia, RICBA19FM: ``Modelli, algoritmi e sistemi per la visualizzazione di grafi e reti''. RG acknowledges support from the Austrian Science Fund (FWF) grant P 31336, SB and MN acknowledge support from FWF grant P 31119.}}
\author{Sujoy Bhore\inst{1,2} 
\and Robert Ganian\inst{1} 
\and Fabrizio Montecchiani\inst{3} 
\and Martin~N\"ollenburg\inst{1} 
}

\institute{Algorithms and Complexity Group, TU Wien, Vienna, Austria\\\email{\{sujoy,rganian,noellenburg\}@ac.tuwien.ac.at}
\and Universit\'{e} libre de Bruxelles (ULB), Bruxelles, Belgium
\and
Dipartimento di Ingegneria, Universit\`{a} degli Studi di Perugia, Perugia, Italy\\\email{fabrizio.montecchiani@unipg.it}
}

\begin{document}

\maketitle

\begin{abstract}
 
An \emph{$h$-queue layout} of a graph $G$ consists of a \emph{linear order} of its vertices and a partition of its edges into $h$ \emph{queues}, such that no two independent edges of the same queue nest. The minimum $h$ such that $G$ admits an $h$-queue layout is the \emph{queue number} of $G$. We present two fixed-parameter tractable algorithms that exploit structural properties of graphs to compute optimal queue layouts. As our first result, we show that deciding whether a graph $G$ has queue number $1$ and computing a corresponding layout is fixed-parameter tractable when parameterized by the treedepth of $G$. Our second result then uses a more restrictive parameter, the vertex cover number, to solve the problem for arbitrary~$h$.

\end{abstract}

\section{Introduction}
An \emph{$h$-queue layout} of a graph $G$ is a linear layout of $G$ consisting of a \emph{linear order} of its vertices and a partition of its edges into \emph{queues}, such that no two independent edges of the same queue nest~\cite{HR92}; see \cref{fig:intro}  for an illustration.  The \emph{queue number} $\qn(G)$ of a graph $G$ is the minimum number of queues in any queue layout of $G$. While such linear layouts represent an abstraction of various problems such as, for instance, sorting and scheduling~\cite{T72,BCLR96}, they also play a central role in three-dimensional graph drawing. It is known that a graph class has bounded queue number if and only if every  graph in this class has a three-dimensional crossing-free straight-line grid drawing in linear volume~\cite{DBLP:journals/comgeo/GiacomoLM05,DBLP:journals/siamcomp/DujmovicMW05}. We refer the reader to~\cite{DBLP:journals/dmtcs/DujmovicW04,Pem92} for further references and applications. Moreover, it is worth recalling that \emph{stack layouts}~\cite{Oll73,Yan89} (or \emph{book embeddings}), which allow nesting edges but forbid edge crossings, form the ``dual'' concept of queue layouts.

\begin{figure}
\centering 
\includegraphics[page=1]{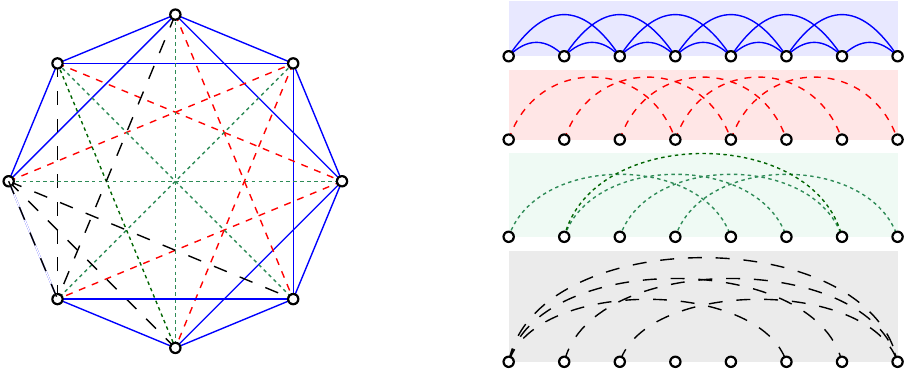}
\caption{A $4$-queue layout of $K_8$.\label{fig:intro}}
\end{figure} 

A rich body of literature is concerned with the study of upper bounds for the queue number of several planar and non-planar graph families (see, e.g.,~\cite{BDDEW18,DBLP:journals/siamcomp/BekosFGMMRU19,BFP13,Duj15,DBLP:conf/focs/DujmovicJMMUW19,DBLP:journals/siamcomp/DujmovicMW05,HLR92,Wie17} and also~\cite{DMW17} for additional references). For instance, a graph of treewidth $w$ has queue number at most $\bigoh(2^w)$~\cite{Wie17}, while every proper minor-closed class of graphs (including planar graphs) has constant queue  number~\cite{DBLP:conf/focs/DujmovicJMMUW19}.

Of particular interest to us is the corresponding recognition problem, which we denote by \qnp: Given a graph $G$ and a positive integer $h$, decide whether $G$ admits an $h$-queue layout. In 1992, in a seminal paper, Heath and Rosenberg proved that $1$-\qnp, i.e., the restriction of \qnp\ to $h=1$, is \NP-complete~\cite{HR92}. In particular, they characterized the graphs that admit queue layouts with only one queue as the arched leveled-planar graphs, and showed that the recognition of these graphs is \NP-complete~\cite{HR92}.

Since \qnp\ is \NP-complete even for a single queue, it is natural to ask under which conditions  the problem can be solved efficiently. For instance, it is known that if the linear order of the vertices is given (and the aim is thus to simply  partition the edges of the graph into queues), then the problem becomes solvable in polynomial time~\cite{HLR92}. We follow up on recent work made for the stack number~\cite{JGAA-526} and initiate the study of the parameterized complexity of \qnp\ by asking under which parameterizations  the problem is fixed-parameter tractable. 
In other words, we are interested in whether ($1$-)\qnp\ can be solved in time $f(k)\cdot n^{\bigoh(1)}$ for some computable function $f$ of the considered structural parameter $k$ of the $n$-vertex input graph~$G$. 

As our main result, we show $1$-\qnp\ is fixed-parameter tractable parameterized by the \emph{treedepth} of the input graph (\cref{sec:td}). We remark that treedepth is a fundamental graph parameter with close ties to the theory of graph sparsity (see, e.g.,~\cite{DBLP:books/daglib/0030491}). The main technique used by the algorithm is iterative pruning, where we recursively identify irrelevant parts of the input and remove these until we obtain a bounded-size equivalent instance (a \emph{kernel}) solvable by brute force. While the iterative pruning technique has already been used in a few other algorithms that exploit treedepth~\cite{DBLP:journals/ai/GanianO18,DBLP:journals/siamdm/GutinJW16,GanianPeitlSlivovskySzeider20}, the unique challenge here lay in establishing that the removal of seemingly irrelevant parts of the graph cannot change NO-instances to YES-instances. The proof of this claim, formalized in Lemma~\ref{lem:main}, uses a new type of block decomposition of $1$-queue layouts. 

For our second result, we turn to the general \qnp\ problem. Here, we establish fixed-parameter tractability when parameterized by a larger parameter, namely the \emph{vertex cover number} (\cref{sec:vc}). This result is also achieved by kernelization and forms a natural counterpart to the recently established fixed-parameter tractability of computing the stack number under the same parameterization~\cite{JGAA-526}, although the technical arguments and steps of the proof differ due to the specific properties of queue layouts.

\smallskip\noindent{\em Note:} Full proofs of statements marked with (*) can be found in the appendix.

\section{Preliminaries}	\label{sec:prel} 
We can assume that our input graphs are connected, as the queue number of a graph is the maximum queue number over all its connected components. Given a graph $G=(V,E)$ and a vertex $v \in V$,  let $N(v)$ be the set of neighbors of $v$ in $G$. Also, for $r \in \mathbb{N}$, we denote by $[r]$ the set $\{1, \ldots, r\}$. An \emph{$h$-queue layout} of $G$ is a pair $\langle \prec, \sigma \rangle$, where $\prec$ is a linear order of $V$, and $\sigma \colon E \rightarrow [h]$ is a function that maps each edge of $E$ to one of $h$ queues. In an $h$-queue layout $\langle \prec, \sigma \rangle$ of $G$, it is required that no two independent edges in the same queue \emph{nest}, that is, for no pair of edges $uv, wx \in E$ with four distinct end-vertices and $\sigma(uv) = \sigma(wx)$, the vertices are ordered as $u \prec w \prec x \prec v$. 
Given two distinct vertices $u$ and $v$ of $G$, $u$ is to the \emph{left} of $v$ if $u \prec v$, else $u$ is to the \emph{right} of $v$. Note that a $1$-queue layout of $G$ is simply defined by a linear order $\prec$ of $V$ and $\sigma \equiv 1$. 

\smallskip We assume familiarity with basic notions in parameterized complexity~\cite{DBLP:series/txcs/DowneyF13,DBLP:books/sp/CyganFKLMPPS15}.

\medskip\noindent \textbf{Treedepth.} Treedepth is a parameter closely related to treewidth, and the structure of graphs
of bounded treedepth is well understood~\cite{DBLP:books/daglib/0030491}.
%
We formalize a few notions needed to define treedepth, see also \cref{fig:treedepth} for an illustration.
A \emph{rooted forest} $\mathcal F$ is a disjoint union of rooted trees.
For a vertex~$x$ in a tree~$T$ 
of $\mathcal F$, the \emph{height} (or {\em depth})
of~$x$ in $\mathcal F$ is the number of vertices in the path from 
the root of~$T$ to~$x$.
The \emph{height of a rooted forest} is the maximum height of a vertex of the forest. 
Let $V(T)$ be the vertex set of any tree $T \in \cal F$.
\begin{definition}[Treedepth]\label{def:td}
  Let the \emph{closure} of a rooted forest~$\cal F$ be the graph
  $\clos({\cal F})=(V_c,E_c)$ with the vertex set 
  $V_c=\bigcup_{T \in \cal F} V(T)$ and the edge set
  $E_c=\{xy \mid \text{$x$ is an ancestor of $y$ in some $T\in\cal F$}\}$.
  A \emph{treedepth decomposition}
  of a graph $G$ is a rooted forest $\cal F$ 
  such that $G \subseteq \clos(\cal F)$.
  The \emph{treedepth} $\td(G)$ of a graph~$G$ is the minimum height of
  any treedepth decomposition of $G$.
\end{definition}

\noindent An optimal treedepth decomposition can be computed by an FPT algorithm.
\begin{proposition}[\cite{ReidlRVS14} ]\label{pro:compute-td}
  Given an $n$-vertex graph $G$ and an integer $k$, it is possible to
  decide whether $G$ has treedepth at most $k$, and if so, to compute an
  optimal treedepth decomposition of $G$ in time $2^{\bigoh(k^2)}\cdot n$.
\end{proposition}

\begin{proposition}[\cite{DBLP:books/daglib/0030491} ]\label{pro:path-td}
Let $G$ be a graph and $\td(G) \le k$. Then~$G$ has no path of length~$2^k$.
\end{proposition}

%

\medskip\noindent \textbf{Vertex cover number.} A \emph{vertex cover} $C$ of a graph $G=(V,E)$ is a subset  $C \subseteq V$ such that each edge in $E$ has at least one incident vertex in $C$. The \emph{vertex cover number} of $G$, denoted by $\tau(G)$, is the size of a minimum vertex cover of~$G$. Observe that $\td(G)\leq \tau(G)+1$: it suffices to build $\cal F$ as a single path with vertex set $C$ and with leaves $V \setminus C$.
Computing an optimal vertex cover of $G$ is FPT.

\begin{proposition}[\cite{DBLP:journals/tcs/ChenKX10} ]\label{pro:compute-vcn}
Given an $n$-vertex graph $G$ and a constant $\tau$, it is possible to decide whether $G$ has vertex cover number at most $\tau$, and if so, to compute a vertex cover $C$ of size $\tau$ of $G$ in time $\bigoh(2^\tau+\tau\cdot n)$.
\end{proposition}

\section{Parameterization by Treedepth}\label{sec:td}

In this section, we establish our main result: the fixed-parameter tractability of $1$-\qnp\ parameterized by treedepth. We formalize the~statement~below.

\begin{theorem}\label{thm:main}
Let $G$ be a graph with $n$ vertices and constant treedepth $k$. We can decide in $\bigoh(n)$ time whether $G$ has queue number one, and, if this is the case, we can also output a $1$-queue layout of $G$. 
\end{theorem}

\subsection{Algorithm Description}
\begin{figure}[t]
\centering
\subfigure[]{\includegraphics[page=1]{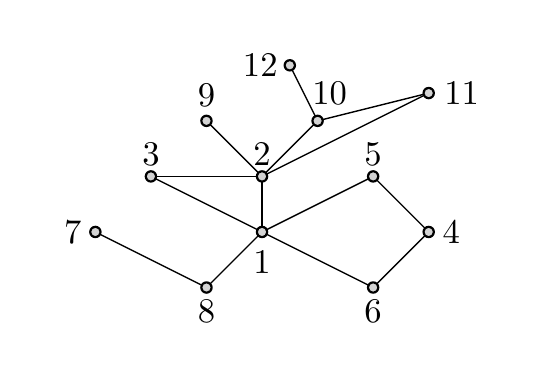}\label{fig:g}}\hfil
\subfigure[]{\includegraphics[page=2]{figs/treedepth}\label{fig:f}}
\caption{(a) A graph $G$ and (b) a treedepth decomposition ${\cal F}$ of $G$ of height $4$. In particular, $P_2=\{1,2\}$, $A_2=\{C_1,C_2,C_3\}$, and $m_2=3$.\label{fig:treedepth}}
\end{figure}

Since we assume $G$ to be connected, any treedepth decomposition of $G$ consists of a single tree $T$. 
Now, suppose that a treedepth decomposition $T$ of $G$ of depth $k$ is given. For a vertex $t$ of $T$, let $P_t$ be the set of ancestors of $t$ including $t$, let $A_t$ be the set of connected components of $G - P_t$ which contain a child of $t$, and $m_t$ be the maximum number of vertices in a component in $A_t$; see also \cref{fig:f}.

\begin{observation}
\label{obs:components}
For every component $C \in A_t$ and for every vertex $v \in C$, it holds that $N(v)\subseteq C \cup P_t$. Moreover, $|C \cup P_t| \le  m_t + k$. 
\end{observation}

Now, we define the following equivalence over components in $A_t$. Components $B,C\in A_t$ satisfy $B\sim C$ if and only if there exists a bijective renaming function $\eta_{B,C}:B\rightarrow C$ over (the vertices of) $B,C$ such that each vertex $b_i\in B$ has a \emph{counterpart} $\eta_{B,C}(b_i) = c_i\in C$ that satisfies: (i) $N(b_i)\cap P_t=N(c_i)\cap P_t$ and (ii) $b_i$ is adjacent to $b_j\in B$ if and only if $c_i$ is adjacent to its counterpart $c_j$. When $B,C$ are clear from the context, we may drop the subscript of $\eta$ for~brevity.

By \cref{obs:components}, the number of equivalence classes of $\sim$ is upper-bounded by the number of possible graphs on $k+m_t$ vertices, which is at most $2^{(k+m_t)^2}$. 
The next observation allows us to propagate the bounds formalized by the notation above from children towards the root.

\begin{observation}
\label{obs:childparent}
If for a vertex $t$ of $T$ there exist integers $a,b$ such that each child $q$ of $t$ satisfies $|A_q|\leq a$ and $m_q\leq b$, then $m_t\leq (a\cdot b)+1$.
\end{observation}

The main component of our treedepth algorithm is \cref{lem:main}, stated below. 
Intuitively, applying \cref{lem:main} bottom-up on $T$ (together with \cref{obs:childparent}) allows us to iteratively remove subtrees from $T$  while preserving the (non-)existence of a hypothetical solution---in particular, we will be able to prune subtrees of parents with a very large number of children until we reach an equivalent instance where each vertex has a bounded number of children.
To formalize the meaning of ``very large'', we define the following function for $i\geq 2$: 
\[\#\emph{children}(k,i)=\big(\big((2^{(k+1)}+1)^{\emph{size}(k,i)^2}+1\big)\cdot (\emph{size}(k,i)+k)!\big)\cdot 2^{(k+\emph{size}(k,i))^2},\]

where $\emph{size}(k,i)$ is a recursively defined function that captures the size bound given by \cref{obs:childparent} as follows:
\begin{itemize}
\item $\emph{size}(k,i)=(\emph{size}(k,i-1)\cdot \#\emph{children}(k,i-1))+1$ for $i \ge 2$, and
\item $\emph{size}(k,1)=\#\emph{children}(k,1)=0$.
\end{itemize}

\begin{lemma}\label{lem:main}
Assume $G$ has a vertex $t$ at depth $i$ in $T$ such that $|A_t|\geq \emph{\#children}(k,i)$, but $m_t\leq \emph{size}(k,i)$ and every descendant $q$ of $t$ in $T$ satisfies that $|A_q|\leq \emph{\#children}(k,i-1)$. Then there exists a component $B$ of $A_t$ such that $G - B$ has queue number one if and only if $G$ has queue number one. Moreover, $B$ can be computed in time $\emph{size}(k,i)!\cdot \emph{\#children}(k,i)^2$.
\end{lemma}

The proof of the lemma is deferred to \cref{sub:proof_lemma1}. Before proceeding, we show how \cref{lem:main} is used to obtain \cref{thm:main}.

\begin{proof}[of \cref{thm:main}]
We start by applying \cref{pro:compute-td} to compute a treedepth decomposition $T$ of $G$ of depth at most $k$. Consider now  vertices at depth $k-1$ in $T$, i.e., vertices whose children are all leaves in $T$, and set $i=2$. Observe that every vertex $v$ at this depth satisfies $m_v\leq \emph{size}(k,2)$ since $\emph{size}(k,2)=1$ and $m_v=1$. If $|A_v|\geq \emph{\#children}(k,2)$, we apply \cref{lem:main} to obtain an equivalent graph with fewer vertices and restart on that graph. Otherwise, every vertex $v$ at depth $k-1$ satisfies $|A_v|<\emph{\#children}(k,2)$.

We now inductively repeat the above argument for every depth less than $k-1$. In particular, assume that for some depth $1\leq d\leq k-1$ 
every vertex $v$ at depth $d$ satisfies $|A_v|<\emph{\#children}(k,i)$, where $i=k-d+1$. Then we can set $d':=d-1$, $i':=i+1$, and recall from \cref{obs:childparent} that every vertex $v'$ at depth $d'$ satisfies $m_v\leq \emph{size}(k,i')$. Hence, if $v'$ has too many subtrees---in particular, if $|A_{v'}|\geq \emph{\#children}(k,i')$---we will once again apply \cref{lem:main} to obtain an equivalent smaller instance, in which case we restart the algorithm. Repeating this procedure for $d'$ will eventually stop, and at that point it will hold that $|A_{v'}|< \emph{\#children}(k,i')$ for every $v'$ at depth $d'$, in turn allowing us to continue with the induction.

The above procedure will halt only once the root $r$ of $T$ satisfies $|A_{r}|< \emph{\#children}(k,k)$ and $m_r\leq \emph{size}(k,k)$. At that point, we have a kernel $G'$~\cite{DBLP:series/txcs/DowneyF13,DBLP:books/sp/CyganFKLMPPS15}---an equivalent graph that has size bounded by a function of $k$, notably by $f(k)=\emph{\#children}(k,k)\cdot \emph{size}(k,k)+1$. To prove \cref{thm:main}, it suffices to decide whether $G'$ admits a $1$-queue layout by a brute-force algorithm that runs in time $\bigoh(f(k)!\cdot f(k)^2)$. Since \cref{lem:main} is applied $\bigoh(n)$ times and the runtime of the associated algorithm is $\bigoh(\emph{size}(k,k)\cdot \emph{\#children}(k,k)^2)$, the total runtime is upper-bounded by a function of $k$ times $n$.
Finally, we note that while it would be possible to provide a term upper-bounding the dependency on $k$ of the running time, it is clear that such a term must necessarily be non-elementary---indeed, the recursive definition of the two functions $\emph{\#children}(k,k)$ and $\emph{size}(k,k)$ results in a tower of exponents of height~$k$.
\qed\end{proof}

\subsection{Proof of \cref{lem:main}}\label{sub:proof_lemma1}
Since we have
$$|A_t|\geq \big(\big((2^{(k+1)}+1)^{\emph{size}(k,i)^2}+1\big)\cdot (\emph{size}(k,i)+k)!\big)\cdot 2^{(k+\emph{size}(k,i))^2}=\#\emph{children}(k,i)$$ 
\noindent and the number of equivalence classes of $\sim$ is upper-bounded by $2^{(k+m_t)^2}\leq 2^{(k+\emph{size}(k,i))^2}$, there must exist an equivalence class, denoted $A_t^\sim\subseteq A_t$, containing at least $\big((2^{(k+1)}+1)^{\emph{size}(k,i)^2}+1\big)\cdot (\emph{size}(k,i)+k)!$ connected components in $A_t$ which are pairwise equivalent w.r.t.\ $\sim$. Moreover, this equivalence class can be computed in time at most $\emph{size}(k,i)!\cdot \#\emph{children}(k,i)^2$ by simply brute-forcing over all potential renaming functions $\eta$ between arbitrarily chosen $\#\emph{children}(k,i)$-many components in $A_t$ to construct the set of all equivalence classes of these components.
Let $B$ be an arbitrarily selected component in $A_t^\sim$.
First, observe that if $G$ is a YES-instance then so is $G - B$, as deleting vertices and edges cannot increase the queue number. On the other hand, assume there is a $1$-queue layout of $G - B$ with linear order $\prec$. Our aim for the rest of the proof is to obtain a linear order $\prec'$ of $G$ that extends $\prec$ and yields a valid $1$-queue layout of $G$.

\paragraph{A Refined Equivalence.}
Let $\equiv_\prec$ be an equivalence over components in $A_t^\sim$ defined as follows. $C \equiv_\prec D$ if and only if the following holds:
the linear order $\prec$ restricted to $P_t\cup \eta_{C,D}(C)$ is the same as $\prec$ restricted to $P_t\cup C$.
In other words, $\equiv_\prec$ is a refinement of $\sim$ restricted to $A_t^\sim$ which groups components based on the order in which their vertices appear (also taking into account which subinterval they appear in w.r.t. $P_t$). Note that $\equiv_\prec$ has at most $(m_t+k)!\leq (\emph{size}(k,i)+k)!$ many equivalence classes, and hence by the virtue of $A_t^\sim$ having size at least $\big((2^{(k+1)}+1)^{\emph{size}(k,i)^2}+1\big)\cdot (\emph{size(k,i)}+k)!$, there must exist an equivalence class $U$ of $\equiv_\prec$ containing at least $(2^{(k+1)}+1)^{\emph{size}(k,i)^2}+1$ components of $A_t^\sim$.

We adopt the following terminology for $U$: we will denote the components in $U$ as $C_1,C_2,\dots,C_u$, where $u=|U|$, we will identify the vertices in a component $C_i$ by using the lower index $i$, and for each such vertex $v$, say $v = v_i\in C_i$, use $v_j$ to denote its counterpart $\eta_{C_i,C_j}(v_i)$.

\paragraph{Identifying Delimiting Components.}
Consider two adjacent vertices $v_i,w_i\in C_i$. We say that component $C_j$ is $vw$-\emph{separate} from $C_i$ if edges $v_iw_i$ and $v_jw_j$ neither nest nor cross each other. On the other hand, $C_j$ is $vw$-\emph{interleaving} (respectively, $vw$-\emph{nesting}) with $C_i$ if $v_iw_i$ and $v_jw_j$ cross each other (respectively, if one of $v_iw_i$ and $v_jw_j$ nests the other). By the definition of $\equiv_\prec$ and $U$, these three cases are exhaustive. Moreover, if $v_iw_i$ is an edge then so is $v_jw_j$ and hence $C_j$ cannot be $vw$-nesting with $C_i$. 
%
Our next aim will be to find two components --- we will call them \emph{delimiting components} --- that are not $vw$-separate for \emph{any} edge $vw$.
To this end, for some two adjacent vertices $v_i,w_i$ of $C_i$, denote by $D_1$ the component whose counterpart to $v_i$ (say $v_1$) is placed leftmost in $\prec$ among all components in $U$. We now define a sequence of components as follows: $D_\ell$ is the unique component that is (i) $vw$-separate from $D_{\ell-1}$ and whose vertex $v_\ell$ is placed (ii) to the right of $v_{\ell-1}$, and (iii) $v_\ell$ is placed leftmost among all components satisfying properties (i) and (ii). Let $d$ be the maximum integer such that $D_d$ exists. 

\newcommand{\dbound}{$d\leq 2^{k+1}+1$.}
\begin{lemma}[*] 
\label{lem:dbound}
\dbound
\end{lemma}

Moreover, each component $C_q$ in $U$ can be uniquely assigned to one component $D_\ell$ as defined above (w.r.t.\ the chosen edge $vw$) as follows: If $C_q=D_\ell$ for some $\ell$, then $C_q$ is assigned to itself; otherwise, $D_\ell$ is the component whose vertex $v_\ell$ is to the left of and simultaneously closest to the corresponding vertex $v_q$ in $C_q$ among all components $D_1,\dots,D_d$.

\newcommand{\interleavetrans}{Let $C_q$ and $C_p$ be two components assigned to the same component $D_\ell$ w.r.t.\ the edge $vw$. Then $C_q$ and $C_p$ are $vw$-interleaving.}
\begin{lemma}[*]
\label{lem:interleavetrans}
\interleavetrans
\end{lemma}

We are now ready to construct our delimiting components. Recall that at this point, $|U|\geq (2^{(k+1)}+1)^{\emph{size}(k,i)^2}+1$ while the maximum number of edges inside a component in $U$ is upper-bounded by $m_t^2\leq \emph{size}(k,i)^2$. Hence by the pigeon-hole principle and by applying the bound provided in \cref{lem:dbound} for each edge inside the components of $U$, there must exist two components in $U$, say $C_x$ and $C_y$, which for each edge $vw$ are assigned to the same component $D^{vw}_\ell$. By \cref{lem:interleavetrans} it now follows that they are $vw$-interleaving for every edge $vw$.

\paragraph{Using Delimiting Components.}
Before we use $C_x$ and $C_y$ to insert $B$, we can show that the way they interleave with each other is ``consistent'' in $\prec$.  

\newcommand{\consistentorder}{Assume, w.l.o.g., that some vertex $v_x$ is to the left of $v_y$. Then for each vertex $w_x$ it holds that $w_x$ is  to the left of $w_y$.}
\begin{lemma}[*]\label{lem:consistent-order}
\consistentorder
\end{lemma}

We remark that it is not the case that $C_x$ must be $vw$-interleaving with $C_y$ if $vw$ is not an edge -- this is, in fact, a major complication that we will need to overcome to complete the proof. 
W.l.o.g. and recalling \cref{lem:consistent-order}, we will hereinafter assume that every vertex $v_x\in C_x$ is placed to the left of its counterpart $v_y \in C_y$. 
The following definition allows us to partition the vertices of $C_x$ into subsequences that should not be interleaved with vertices of $B$. 

\begin{definition}[Block]\label{def:block}
A \emph{block} $L=\{v_x^1, v_x^2 , \dots, v_x^h\}$ of $C_x$ is a maximal set of vertices of $C_x$ such that: (1) there is no vertex $v_y^i$ (the counterpart in $C_y$ of $v_x^i$), with $1 \le i \le h$, between two vertices of $L$ in $\prec$; (2) there are no two vertices of $L$ such that one has a neighbor to its left and one has a neighbor to its right. 
\end{definition}

\noindent We observe that, as an immediate consequence of \cref{def:block}, no two vertices of $L$ are adjacent (an edge $uv$ in $L$ would imply that $u$ has a neighbor to its right and $v$ has a neighbor to its left, or vice versa).

For each block $L=\{v_x^1, v_x^2 , \dots, v_x^h\}$ of $C_x$, there is a corresponding set of vertices $\{v_B^1, v_B^2, \dots, v_B^h \}$ of $B$, i.e., the set containing the counterparts of $L$ in $B$. We will obtain a linear order of $G$ by processing the blocks of $C_x$ one by one as encountered in a left-to-right sweep of $\prec$, and for each block $L$, we will extend $\prec$ by suitably inserting  the corresponding vertices of $B$.

\begin{figure}[t]
\centering
\subfigure[]{\includegraphics{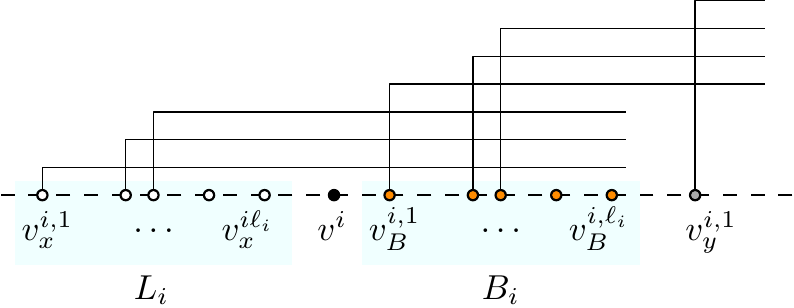}}
\subfigure[]{\includegraphics[width=\columnwidth]{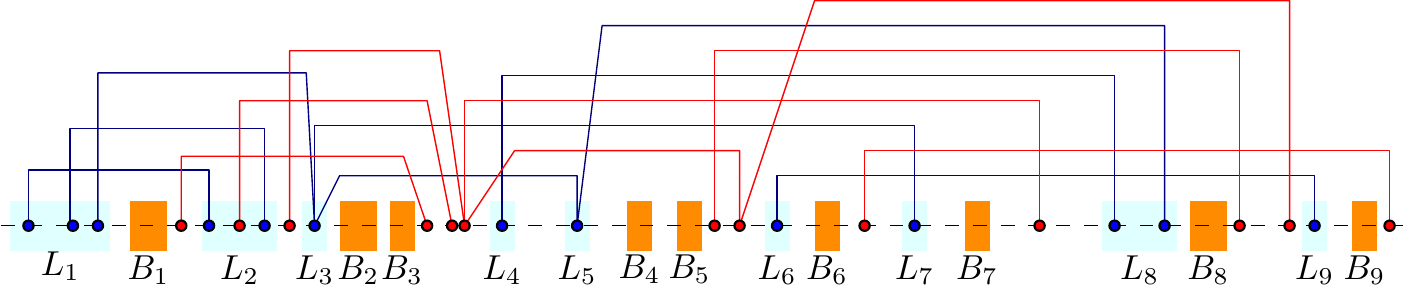}}
\caption{Reinsertion of $B_i$: (a) A schematic illustration, and (b) an example where blue and red vertices belong to $C_x$ and $C_y$, respectively.\label{fig:ordering}}
\end{figure}

Consider the $i$-th encountered block $L_i=\{v_x^{i,1}, v_x^{i,2} , \dots, v_x^{i,\ell_i}\}$ of $C_x$, refer to \cref{fig:ordering} for an illustration. Note that, because $C_x$ and $C_y$ are equivalent components, it holds $v_y^{i,1} \prec v_y^{i,2} \dots \prec v_y^{i,\ell_i}$ (even though such vertices might not be consecutive). Also, let $v^i$ be the first vertex to the left of $v_y^{i,1}$ in $\prec$ (possibly $v^i = v_x^{i,\ell_i}$). We insert all vertices in the corresponding block $B_i$ of $B$ such that: $v_i \prec v_B^{i,1} \prec v_B^{i,2} \prec \dots v_B^{i,\ell_i} \prec v^{i,1}_y$. After processing the last block of $C_x$, we know that all vertices of $C_x$ have been considered and hence all vertices of $B$ have been reinserted, that is, we extended $\prec$ to a linear order $\prec'$ of the whole graph $G$. The next observation immediately follows by the procedure described above.

\begin{observation}\label{obs:relative}
For every vertex $v_x$, it holds that $v_x \prec' v_B \prec' v_y$.
\end{observation}

We now establish the correctness of $\prec'$, completing the proof of Lemma~\ref{lem:main}.

\newcommand{\lemprec}{The linear order $\prec'$ yields a valid $1$-queue layout of $G$.}
\begin{lemma}[*]\label{lem:prec}
\lemprec
\end{lemma}
\begin{proof}[sketch]
\begin{figure}[p]
\centering
\subfigure[]{\includegraphics[page=1]{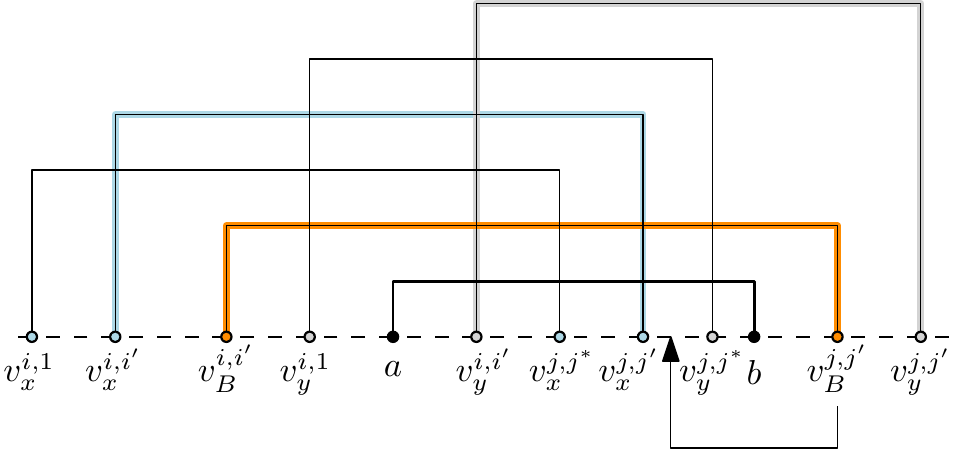}\label{fig:prec-a}}
\subfigure[]{\includegraphics[page=2]{figs/prec}\label{fig:prec-b}}
\subfigure[]{\includegraphics[page=3]{figs/prec}\label{fig:prec-c}}
\caption{Illustration for the proof of \cref{lem:prec}: $v_B^{i,i'}v_B^{j,j'}$ nests an edge $ab$.}
\end{figure}
To prove the statement, we argue that no two edges of $G$ nest in the $1$-queue layout defined by $\prec'$. We recall that $\prec'$ extends $\prec$, hence we do not need to argue about pairs of edges in $G - B$.
Moreover, by construction, $\prec'$ restricted to $C_x$ is the same as $\prec'$ restricted to $B$ (up to the renaming function $\eta$). Consequently, no two edges having both endpoints in $B$ can nest. 

We first consider any edge $v_Bw$ for $w \in P_t$ and $v_B\in B$, and assume $v_B \prec' w$ (else the argument is symmetric). Suppose, for a contradiction, that $v_Bw$ nests another edge $ab$. Recall that since $C_x$ and $B$ are equivalent components, if $v_B$ is to the left of $w$, the same holds for $v_x$. By \cref{obs:relative}, we know $v_x \prec' v_B \prec' w$, which implies that $ab$ is nested by $v_xw$ as well, a contradiction with  the correctness of $\prec$.
Similarly, if $v_Bw$ is nested by an edge $ab$, then we know $v_B \prec' v_y \prec' w$, which implies that $ab$ nests $v_yw$ as well, again a contradiction.

We now consider any edge $v_Bw_B$, with $v_B \prec' w_B$, and we assume for a contradiction that $v_Bw_B$ nests an edge $ab$. Since Definition~\ref{def:block} ensures that a block cannot contain a pair of adjacent vertices, we know that $v_x$ and $w_x$ belong to different blocks, say $L_i$ and $L_j$ (with $i < j$) respectively. Therefore, we can rename the vertices as $v_x = v_x^{i,i'}$ and $w_x = v_x^{j,j'}$, and similarly $v_B = v_B^{i,i'}$ and $w_B = v_B^{j,j'}$; refer to \cref{fig:prec-a} for an illustration. By \cref{obs:relative}, it holds $v_x^{i,i'} \prec' v_B^{i,i'} \prec' v_y^{i,i'}$ and $v_x^{j,j'} \prec' v_B^{j,j'} \prec' v_y^{j,j'}$. Moreover, the correctness of $\prec$ implies that $v_B^{i,i'} \prec' a \prec' v_y^{i,i'}$ (since $v_y^{i,i'}v_y^{j,j'}$ cannot nest $ab$) 
and $v_x^{j,j'}\prec' b \prec' v_B^{j,j'}$ (since $v_x^{i,i'}v_x^{j,j'}$ cannot nest $ab$). 
Because $a$ is between $v_B^{i,i'}$ and $v_y^{i,i'}$, either there exists another vertex $v_y^{i,1}$ (the counterpart to the first vertex in block $L_i$, where possibly $v_y^{i,1}=a$)  such that $v_B^{i,i'} \prec' v_y^{i,1} \preceq' a \prec' v_y^{i,i'}$, or $a = v_y^{i,i'}$.

Suppose first $a \neq v_y^{i,1}$ and $a \neq v_y^{i,i'}$. Observe that $v_x^{i,1}$ has at least one neighbor in $C_x$ (because $C_x$ is connected), and that $v_x^{j,j'}$ is to the right of $v_x^{i,i'}$, hence, by \cref{def:block}, $v_x^{i,1}$ also has a neighbor to its right, say $v_x^{l,j^*}$. Because no two edges nest in $\prec$, it must be: (i) $v_x^{i,1} \prec' v_x^{i,i'}$, (ii) $v_x^{l,j^*} \prec' b$, and (iii) $v_y^{l,j^*} \prec' b$ (possibly $v_y^{l,j^*} = b$). 
 Altogether, this implies that $v_x^{j,j'}$ and $v_x^{l,j^*}$ are in the same block (i.e., $l=j$) and hence $v_B^{j,j'} \prec' v_y^{j,j^*} \prec' b$, which contradicts $b \prec' v_B^{j,j'}$. If instead  $a = v_y^{i,1}$ or $a = v_y^{i,i'}$, then $b$ is either a vertex of $C_y$ or a vertex of $P_t$. If $b \in C_y$, the argument is similar, as we can set $b=v_y^{j,j^*}$ and observe that $v_B^{j,j'}$ should be to the left of $v_y^{j,j^*}$, see \cref{fig:prec-b}. If $b \in P_t$, we would have $v_x^{j,j'} \prec' b \prec' v_y^{j,j'}$, which contradicts the fact that $C_x$ and $C_y$ are equivalent components, see \cref{fig:prec-c}. 
\qed\end{proof}

\section{Parameterization by Vertex Cover Number}\label{sec:vc}

We now turn to the general \qnp problem and show that it is fixed-parameter tractable when parameterized by the vertex cover number by proving:

\begin{theorem}\label{th:vc-main}
Let $G$ be a graph with $n$ vertices and vertex cover number $\tau=\tau(G)$. A queue layout of $G$ with the minimum number of queues can be computed in $\bigoh(2^{{\tau^{\bigoh(\tau)}}}+\tau \log \tau \cdot n)$ time.
\end{theorem}

\subsection{Algorithm Description}

Before describing the algorithm behind \cref{th:vc-main}, we make an easy observation (which matches an analogous observation in~\cite{JGAA-526}).

\begin{lemma}\label{lem:vc}
Every $n$-vertex graph $G=(V,E)$ with a vertex cover $C$ of size $\tau$ admits a $\tau$-queue layout. 
Moreover, if $G$ and $C$ are given as input, such a $\tau$-queue layout can be computed in $\bigoh(n+\tau\cdot n)$ time.
\end{lemma}
\begin{proof}
Denote by $c_1,\ldots,c_\tau$ the $\tau$ vertices of $C$ and let $\prec$ be any linear order of $G$ such that $c_i \prec c_{i+1}$, for $i=1,2,\dots,\tau-1$. A queue assignment $\sigma$ of $G$ on $h$ queues can be obtained as follows.
Let $U=V\setminus C$. For each $i\in [\tau]$ all edges $u c_i$ with $u \in U\cup \{c_1,\dots,c_{i-1}\}$ are assigned to queue $i$.  
Now, consider the edges assigned to any queue $i \in [\tau]$. 
By construction, they are all incident to vertex $c_i$, and thus no two of them nest each other. 
Therefore, the pair $\langle \prec, \sigma \rangle$ is a $\tau$-queue layout of $G$ and can be computed in $\bigoh(n+\tau\cdot n)$ time.\qed\end{proof}

Let $C$ be a vertex cover of size $\tau$ of graph $G$. For any subset $U$ of $C$, a vertex  $v \in V \setminus C$ is of \emph{type $U$} if $N(v) = U$. This defines an equivalence relation on $V \setminus C$ and in particular partitions $V \setminus C$ into at most $\sum_{i=1}^{\tau} {\tau \choose{i}}=2^{\tau-1} < 2^\tau$ distinct types. Denote by $V_U$ the set of vertices of type $U$. 

\begin{lemma}\label{le:kernel}
Let $h \in \mathbb N$ and $v \in V_U$ such that $|V_U| \ge 2 \cdot h^\tau + 2$.
Then $G$ admits an $h$-queue layout if and only if $G'=G - \{v\}$ does. Moreover, an $h$-queue layout of $G'$ can be extended to an $h$-queue layout of $G$ in linear time.
\end{lemma}

The proof of \cref{le:kernel} is deferred to \cref{ss:lekernel}.

\begin{proof}[of \cref{th:vc-main}]
By \cref{pro:compute-vcn}, we can determine the vertex cover number $\tau$ of $G$ and compute a vertex cover $C$ of size $\tau$ in time $\bigoh(2^\tau+\tau\cdot n)$. 
With \cref{le:kernel} in hand, we can then apply a binary search on the number of queues $h \le \tau$ as follows. 
If $h > \tau$, by \cref{lem:vc} we can immediately conclude that $G$ admits a $\tau$-queue layout and compute one in $\bigoh(n+\tau\cdot n)$ time. Hence we shall assume that $h \le \tau$. We construct a kernel $G^*$ from $G$ of size $h^{\bigoh(\tau)}$ as follows. We first classify each vertex of $G$ based on its type. We then remove an arbitrary vertex from each set $V_U$ with $|V_U| > 2 \cdot h^\tau + 1$ until $|V_U| \le 2 \cdot h^\tau + 1$. Thus, constructing $G^*$ can be done in $\bigoh(2^\tau + \tau \cdot n)$ time, since $2^\tau$ is the number of types and $\tau \cdot n$ is the maximum number of edges of $G$.
From \cref{le:kernel} we conclude that $G$ admits an $h$-queue layout if and only if $G^*$ does. 

Given a linear order $\prec^*$ of $G^*$, a queue assignment $\sigma^*$ such that $\langle \prec^*, \sigma^* \rangle$ is an $h$-queue layout of $G^*$ exists if and only if $\sigma^*$ contains no $h$-rainbow~\cite{HLR92}, i.e., $h$ independent edges that  pairwise nest, which can be easily checked (and computed if it exists) in $h^{\bigoh(\tau)}$ time~\cite{HLR92}. Consequently, determining whether $G^*$ admits an $h$-queue layout can be done by first guessing all linear orders, and then for each of them by testing for the existence of an $h$-rainbow. Since we have $2^\tau$ types, and each of the at most $2 \cdot h^\tau + 1$ elements of the same type are equivalent in the queue layout (that is, the position of two elements of the same type can be exchanged in $\prec^*$ without affecting $\sigma^*$), the number of linear orders can be upper bounded by $(2^\tau)^{\bigoh(h^\tau)}=2^{\tau^{\bigoh(\tau)}}$. Thus, whether $h$ queues suffice for $G^*$ can be determined in $2^{\tau^{\bigoh(\tau)}} \cdot h^{\bigoh(\tau)}=2^{\tau^{O(\tau)}}$ time. 
An $h$-queue layout of $G^*$ (if any) can be extended to one of $G$ by iteratively applying the constructive procedure of \cref{le:kernel}, in $\bigoh(\tau \cdot n)$ time.
Finally, by applying a binary search on $h$ we obtain an overall time complexity of $\bigoh(2^{{\tau^{\bigoh(\tau)}}}+\tau \log \tau \cdot n)$, as desired.
\qed\end{proof}

\subsection{Proof of \cref{le:kernel}}\label{ss:lekernel}
One direction follows easily, since removing a vertex from an $h$-queue layout still gives an $h$-queue-layout of the resulting graph. So let $\langle \prec, \sigma \rangle$ be an $h$-queue layout of $G'$. 
We prove that an $h$-queue layout of $G$ can be constructed by inserting $v$ immediately to the right of a suitable vertex $u$ in $V_U$ and by assigning the edges of $v$ to the same queues as the corresponding edges of $u$. 

We say that two vertices $u_1, u_2 \in V_U$ are \emph{queue equivalent}, if for each vertex $w \in U$, the edges $u_1w$ and $u_2w$ are both assigned to the same queue according to $\sigma$. 
Each vertex in $V_U$ has degree exactly $|U|$, hence this relation partitions the vertices of $V_U$ into at most 
$h^{|U|} \le h^\tau$ sets. 
Let $V^*_U = V_U \setminus \{v\}$. Since $|V^*_U| \ge 2 \cdot h^\tau + 1$, at least three vertices of this set, which we denote by $u_1$, $u_2$, and $u_3$,  are queue equivalent. 
Consider now the graph induced by the edges of these three vertices that are assigned to a particular queue. 
By the above argument, such a graph is a $K_{l,3}$, for some $l>0$. However, $K_{3,3}$ does not admit a $1$-page queue layout, because any graph with queue number $1$ is planar~\cite{HR92}. As a consequence, $l \le 2$, that is, each $u_i \in V^*_U $ has at most two edges on each queue. Denote such two edges by $u_iw$ and $u_iz$ and assume, w.l.o.g., that $u_1 \prec u_2 \prec u_3$ and $w \prec z$. We now claim that $w \prec u_1 \prec u_2 \prec u_3 \prec z$, else two edges would nest. We can  distinguish a few cases based on the position of $u_1$ (recall that $u_1 \prec u_2 \prec u_3$), refer to \cref{fig:cases} for an illustration.

\begin{figure}[h!]
\centering
\subfigure[Case A]{\includegraphics[page=1]{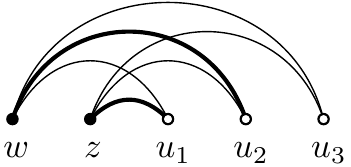}}\hfill
\subfigure[Case B.1]{\includegraphics[page=2]{figs/cases}}\hfill
\subfigure[Case B.2]{\includegraphics[page=3]{figs/cases}}
\subfigure[Case B.3]{\includegraphics[page=4]{figs/cases}}\hfill
\subfigure[Case C.1]{\includegraphics[page=5]{figs/cases}}\hfill
\subfigure[Case C.2]{\includegraphics[page=6]{figs/cases}}
\caption{Illustration for the proof of \cref{le:kernel}.}\label{fig:cases}
\end{figure}
 
\begin{itemize}

\item \textbf{Case A:} $w \prec z \prec u_1$, then the nesting edges are $z u_1$ and $w u_2$.

\item \textbf{Case B:} $u_1 \prec w \prec z$, then we distinguish three more subcases. 

\begin{itemize}
\item \textbf{Case B.1:} $u_2 \prec w$, then the nesting edges are $u_1 z$ and $u_2 w$. 
\item \textbf{Case B.2:} $w \prec u_2 \prec z$, then the nesting edges are $u_1 z$ and $w u_2$. 
\item \textbf{Case B.3:} $z \prec u_2$, then the nesting edges are $z u_2$ and $w u_3$. 
\end{itemize}

\item \textbf{Case C:} $w \prec u_1 \prec z$, if $w \prec u_2 \prec u_3 \prec z$ the claim follows. Else, we have two more subcases based again on the position of $u_2$.

\begin{itemize}
\item \textbf{Case C.1:} $w \prec z \prec u_2$, then the nesting edges are $w u_2$ and  $u_1 z$.
\item \textbf{Case C.2:} $w \prec u_2  \prec z \prec u_3$, then the nesting edges are $w u_3$ and $u_1 z$. 
\end{itemize}

\end{itemize}

\noindent It follows that we can extend $\prec$ by introducing $v$ as the first vertex to the right of $u_1$ and, for each edge $vw$ such that $w \in U$, we can assign $vw$ to the same queue as $u_1w$. This operation does not introduce any nesting. Namely, if $vw$ is assigned to a queue containing only one edge of $u_1$, the graph induced by the edges in this queue is a star with center $w$ and no two edges can nest. If  $vw$ is assigned to a queue containing two edges of $u_1$, say $u_1 w$ and $u_1 z$, then we know that all vertices of $V_U$ are between $w$ and $z$ in $\prec$ and again no two edges nest. 

\section{Conclusions and Open Problems}\label{sec:conclusions}

We proved that $h$-\qnp is fixed-parameter tractable parameterized by treedepth for $h=1$, and by the vertex cover number for arbitrary $h \ge 1$. Several interesting questions arise from our research, among them: 

\begin{enumerate}
\item A first natural question is to understand whether \cref{thm:main} can be extended to the general case ($h \ge 1$). In particular, our arguments establishing the existence of interleaving components already fail for $h=2$.
\item Extending \cref{thm:main} to graphs of bounded treewidth is also an interesting problem; here the main issue is to be able to forget information about vertices in a partial order, thus an approach based on testing arched leveled-planarity might be more suitable. 

\item Finally, we mention the possibility of studying the parameterized complexity of \emph{mixed} linear layouts, using both queues and stacks, see~\cite{DBLP:conf/gd/ColKN19,DBLP:journals/dmtcs/DujmovicW05,HR92,DBLP:conf/gd/Pupyrev17}.

\end{enumerate} 

\clearpage

\bibliographystyle{splncs04}
\bibliography{GD-ref}
\clearpage

\appendix

\section*{Appendix}

\section{Missing Proofs for \cref{sec:td}}\label{ap:td}

\setcounter{lemma}{1}
\begin{lemma}
\dbound
\end{lemma}
\begin{proof}
Consider for a contradiction that there exists a component $D_\ell$ such that $\ell>2^k$ and $\ell<d-2^k$, i.e., that there is a sequence of pairwise $vw$-separated components to the left as well as to the right of $D_\ell$. By the connectivity of $G$, there must be a path from $v$ to some vertex in $P_t$, say $p$. However, by the definition of $\equiv_\prec$ every vertex in $P_t$ lies either to the left of $v_1$ or to the right of $w_d$, and hence a path from $v$ to $p$ would need to pass through a sequence of $2^k$ edges forming disjoint intervals in the linear order $\prec$. Since nestings are not allowed, such a path must have at least one vertex inside each of these intervals, and hence its length is at least $2^k$, which contradicts \cref{pro:path-td}.
\end{proof}

\begin{lemma}
\interleavetrans
\end{lemma}
\begin{proof}
Assume w.l.o.g.\ that $w_\ell$ is placed to the right of $v_\ell$. Since both $C_q$ and $C_p$ are assigned to $D_\ell$, the counterparts $w_q$ and $w_p$ to $w_\ell$ must be placed to the right of $w_\ell$ while the counterparts $v_q$ and $v_p$ to $v_\ell$ must be placed to the left of $w_\ell$. Hence $C_q$ and $C_p$ cannot be $vw$-separate, and the observation follows by recalling that $C_q$ and $C_p$ cannot be $vw$-nesting either. 
\end{proof}

\begin{lemma}
\consistentorder
\end{lemma}
\begin{proof}
Consider for a contradiction that there is a vertex $w_x$ to the right of $w_y$. Consider a $v_x$-$w_x$ path $P_x$ in $G[C_x]$, and let $P_y$ be the $v_y$-$w_y$ path in $G[C_y]$ consisting of the counterparts of $P_x$. Let $a_xb_x$ be the first edge on $P_x$ such that $a_x$ is placed to the left of $a_y$ but $b_x$ is placed to the right of $b_y$. Then the edges $a_xb_x$ and $a_yb_y$ would be nesting, contradicting the correctness of $\prec$.
\end{proof}

\begin{lemma}
\lemprec
\end{lemma}
\begin{proof}[missing part]
\begin{figure}[p]
\centering
\subfigure[]{\includegraphics[page=4,scale=0.9]{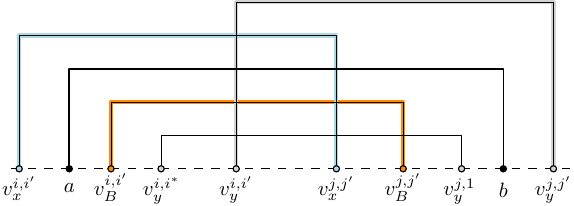}\label{fig:prec-d}}
\subfigure[]{\includegraphics[page=5,scale=0.9]{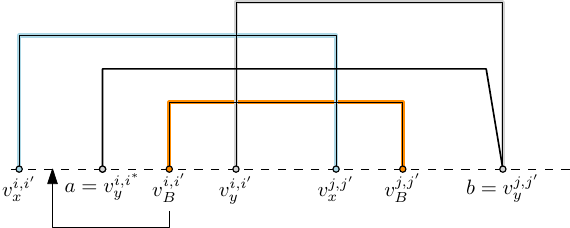}\label{fig:prec-e}}
\subfigure[]{\includegraphics[page=6,scale=0.9]{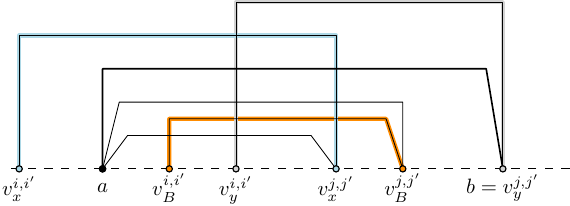}\label{fig:prec-f}}
\caption{Illustration for the proof of \cref{lem:prec}: $v_B^{i,i'}v_B^{j,j'}$ is nested by an edge $ab$.}
\end{figure}
To conclude the proof of the lemma, we consider any edge $v_Bw_B$, with $v_B \prec' w_B$, and we assume for a contradiction that $v_Bw_B$ is nested by an edge $ab$. Again we can rename the vertices as $v_x = v_x^{i,i'}$ and $w_x = v_x^{j,j'}$, and similarly $v_B = v_B^{i,i'}$ and $w_B = v_B^{j,j'}$. By the position of $b$ we can deduce either that $b= v_y^{j,j'}$ (possibly $j'=1$) or that  edge $v_y^{i,i^*}v_y^{j,1}$ exists. In the latter case either $v_y^{i,i^*}v_y^{j,1}$ is also nested by $ab$ or $v_B^{i,i'} \prec' a$, and in both cases we obtain a contradiction; refer to \cref{fig:prec-d} for an illustration. In the former case, we should again distinguish whether $a \in C_y$ or $a \in P_t$. If $a \in C_y$, it should be $v_B^{i,i'} \prec a=v_y^{i,i^*}$, see \cref{fig:prec-e}. If $a \in P_t$, we would have $v_x^{i,i'} \prec' a \prec' v_y^{i,i'}$, which again contradicts the fact that $C_x$ and $C_y$ are equivalent components, see \cref{fig:prec-f}.
\end{proof}

\end{document}